\documentclass[12pt,reqno,a4paper]{amsart}
\usepackage{amsmath,amssymb,dsfont,graphicx,cite,latexsym,
epsf,cancel,epic,eepic}
\usepackage{amssymb,mathrsfs,txfonts}
\usepackage[colorlinks=false]{hyperref}
\usepackage{mathpazo}
\usepackage{color}
\newtheorem{thm}{Theorem}[section]
\newtheorem{prop}[thm]{Proposition}
\newtheorem{cor}[thm]{Corollary}

\theoremstyle{definition}

\newtheorem{ex}[thm]{Example}

\newcommand{\R}{\mathbb R}

\newcommand{\Q}{\mathbb Q}

\newcommand{\x}{\widetilde x}
\newcommand{\T}{\mathrm T}
\newcommand{\e}{\mathfrak e}
\newcommand{\Hxx}{\undertilde{\mathfrak H}}
\newcommand{\HXX}{\widetilde{\mathfrak H}}
\newcommand{\HXx}{\mathfrak H}
\def\t{\widetilde}
\def\ep{\varepsilon}
\def\epsilon{\varepsilon}
\def\beq{\begin{equation}}
\def\eeq{\end{equation}}
\def\bea{\begin{eqnarray}}
\def\eea{\end{eqnarray}}
\def\nn{\nonumber}
\newbox\meibox
\def\placeunder#1#2#3#4{\setbox\meibox%
\vbox{\hbox{\hskip#4$\hphantom{#2}$}\hbox{$\hphantom{#1}$}}%
\vtop{\baselineskip=0pt\lineskiplimit=\baselineskip%
\lineskip=#3\hbox to \wd\meibox{\hfil\hskip#4$#2$\hfil}%
\hbox to \wd\meibox{\hfil$#1$\hfil}}}
\def\undertilde#1{\mathchoice{%
\placeunder{\vbox to 1.4pt{\hbox{$\displaystyle\widetilde{\,\,\,
}$}\vss}}{\displaystyle#1}{1.5pt}{1.5pt}}%
{\placeunder{\vbox to 1.4pt{\hbox{$\textstyle\widetilde{\,\,
}$}\vss}}{\textstyle#1}{1.5pt}{1.5pt}}%
{\placeunder{\vbox to 1.4pt{\hbox{$\scriptstyle\tilde{
}$}\vss}}{\scriptstyle#1}{1pt}{1pt}}%
{\placeunder{\vbox to 1.4pt{\hbox{$\scriptscriptstyle\tilde{
}$}\vss}}{\scriptscriptstyle#1}{1pt}{1pt}}%
}

\title[Quadratic vector fields admitting integral-preserving discretizations]
{New classes of quadratic vector fields admitting integral-preserving Kahan-Hirota-Kimura discretizations}

\author{Matteo Petrera and Ren\'e Zander}

\begin{document}

\maketitle

\begin{center}
{\footnotesize{
Institut f\"ur Mathematik, MA 7-2\\
Technische Universit\"at Berlin, Str. des 17. Juni 136,
10623 Berlin, Germany
}}
\end{center}

\begin{abstract}
\noindent
We present some new  families of quadratic vector fields, not  necessarily
integrable, for which their Kahan-Hirota-Kimura discretization exhibits the preservation of
some of the characterizing features of the underlying continuous systems (conserved quantities and invariant measures).
\end{abstract}

\let\thefootnote\relax\footnote{E-mail: petrera@math.tu-berlin.de, zander@math.tu-berlin.de}

\section{Introduction}

The purpose of this paper is to present some new classes of quadratic vector fields for which the Kahan-Hirota-Kimura discretization \cite{K} produces birational maps which preserve some of the characterizing features of the underlying continuous systems, 
as for example conserved quantities and invariant measures.

It is a well-known fact in numerical integration that 
if a dynamical system
has conserved quantities, or is volume-preserving, or has some other important geometrical feature
(such as being invariant under the action of a Lie group of symmetries) a generic discretization scheme of the underlying differential equations does not
share the same properties. On the contrary, Kahan-Hirota-Kimura discretizations, which are
applicable whenever the continuous vector field is quadratic, seem to inherit several good properties
of the continuous systems they are discretizing. In particular, for the special case
of completely integrable systems, ideally one would like to obtain discretizations which are
themselves completely integrable. As a matter of fact
the Kahan-Hirota-Kimura discretization of many known integrable quadratic systems of differential equations possesses
this remarkable integrability-preserving feature  \cite{PPS1, PPS2,PPS3,HP, C1,C2,C3,PS1,PS2,PS3,PS4,Ko}.

The Kahan-Hirota-Kimura discretization scheme has been introduced in 1993 by W. Kahan in the unpublished notes\cite{K}. It can be applied to any system of ordinary differential equations $\dot{x}=f(x)$ for $x:\R \to \R^n$
with
\begin{equation}
\label{vectorfield}
f(x)=Q(x)+Bx+c,\quad x\in\R^n.
\end{equation}
Here each component of $Q:\R^n \to \R^n$ is a quadratic form, while $B\in {\rm{Mat}}_{n \times n} (\R)$ and
$c \in \R^n$.
Then the Kahan-Hirota-Kimura discretization is given by
\begin{equation}
\label{KHKdiscretization}
\frac{\x-x}{2\epsilon}=Q(x,\x)+\frac12B(x+\x)+c,
\end{equation}
where 
\begin{equation*}
Q(x,\x)=\frac12\left(Q(x+\x)-Q(x)-Q(\x)\right),
\end{equation*}
is the symmetric bilinear form corresponding to the quadratic form
$Q$. Here and below we use the following notational convention
which will allow us to omit a lot of indices: for a sequence
${x}:\mathbb Z\to\mathbb R$ we write ${x}$ for ${x}_k$ and $\widetilde{{x}}$ for
${x}_{k+1}$. Equation (\ref{KHKdiscretization}) is {linear} with respect
to $\widetilde {x}$ and therefore defines a {rational} map
$\widetilde{{x}}=\Phi({x},\epsilon)$. Clearly, this map approximates the
time-$(2\epsilon)$-shift along the solutions of the original
differential system. 
(We have
chosen a slightly unusual notation $2\epsilon$ for the time step,
in order to avoid appearance of powers of 2 in numerous
formulae; a more standard choice would lead to changing
$\epsilon\mapsto\epsilon/2$ everywhere.) 
Since equation (\ref{KHKdiscretization}) remains invariant under the interchange
${x}\leftrightarrow\widetilde{{x}}$ with the simultaneous sign
inversion $\epsilon\mapsto-\epsilon$, one has the {
reversibility} property $
\Phi^{-1}({x},\epsilon)=\Phi({ x},-\epsilon).
$
In particular, the map $\Phi$ is {birational}.
The explicit form of the map $\Phi$ defined by \eqref{KHKdiscretization} is 
\begin{equation}
\t x =\Phi(x,\ep)= x + 2\ep \left( I - \ep f'(x) \right)^{-1} f(x),
\label{eq: Phi gen}
\end{equation}
where $f'(x)$ denotes the Jacobi matrix of $f(x)$. Moreover, if the vector field $f(x)$ is homogeneous (of degree 2), then \eqref{eq: Phi gen} can be equivalently rewritten as
\begin{equation}
\t x =\Phi(x,\ep)= \left( I - \epsilon f'(x) \right)^{-1} x.
\label{eq: Phi gen2}
\end{equation}
Due to the reversibility of the map, in the latter case we also have:
$$
x =\Phi(\t x,-\ep)= \left( I + \epsilon f'(\t x) \right)^{-1} \t x \quad \Leftrightarrow \quad \t x = \left( I + \epsilon f'(\t x) \right) x.
$$

Kahan applied this discretization scheme to the famous
Lotka-Volterra system and showed that in this case it possesses a
very remarkable non-spiralling property. 

The next, definitely more intriguing, appearance of this discretization
was in the two papers by R. Hirota and K. Kimura who (being
apparently unaware of the work by Kahan) applied it to two famous
{integrable} systems of classical mechanics, the Euler top and
the Lagrange top \cite{HK1, HK2}. Surprisingly, the Kahan-Hirota-Kimura
discretization scheme produced {integrable  maps} 
in both the Euler and the Lagrange
cases of rigid body motion. 
Even more
surprisingly, the mechanism which assures integrability in these
two cases seems to be rather different from the majority of
examples known in the area of integrable discretizations, and,
more generally, integrable maps, cf. \cite{S}. 

In the last years many efforts in understanding the
integrability mechanism of the Kahan-Hirota-Kimura
scheme have been spent. We refer to the papers
\cite{PPS1, PPS2,PPS3,HP, C1,C2,C3,PS1,PS2,PS3,PS4,Ko}
for a rich collection of results, including many
examples of integrability-preserving Kahan-Hirota-Kimura
maps.


\section{Higher-dimensional generalizations of reduced Nahm systems}

The first class of quadratic differential equations we want to consider is a higher-dimensional generalization of the two-dimensional Nahm systems introduced in \cite{HMM},
\begin{equation} \label{k3}
\left\{ \begin{array}{l}
\dot x_1 = x_1^2 - x_2^2, \vspace{.1truecm}\\
\dot x_2 = - 2x_1x_2,
\end{array} \right.
\qquad
\left\{ \begin{array}{l}
\dot x_1 = {\displaystyle{2 x_1^2 - 12x_2^2}}, \vspace{.1truecm}\\
\dot x_2 = - 6 x_1x_2 - 4 x_2^2,
\end{array} \right. \qquad 
\left\{ \begin{array}{l}
\dot x_1 = {\displaystyle{2 x_1^2 - x_2^2}}, \vspace{.1truecm}\\
\dot x_2 = - 10 x_1x_2+ x_2^2.
\end{array} \right.
\end{equation}
Such systems can be explicitly integrated in terms of elliptic functions and they admit
integrals of motion given respectively by
\begin{eqnarray}
&&H_1= x_2 (3 x_1^2 - x_2^2),\label{H1} \\
&&H_2= x_2(2x_1+3x_2)(x_1-x_2)^2, \\
&&H_3= x_2(3x_1 -x_2)^2(4x_1 +x_2)^3.\label{H3} 
\end{eqnarray}
Systems (\ref{k3}) have been discussed in \cite{HMM} and
discretized by means of the 
Kahan-Hirota-Kimura scheme in \cite{PPS2} (see also \cite{CT}). There it is shown that the corresponding Kahan-Hirota-Kimura maps
admit conserved quantities which are $\ep$-perturbations of the continuous ones given by (\ref{H1}--\ref{H3}).

We now observe that
 systems (\ref{k3}) are  special two-dimensional instances of the following parametric $n$-dimensional family of quadratic differential equations:
\begin{equation}
\dot x =\ell_1^{1-a}\ell_2^{1-b}\ell_3^{1-c}J\nabla(\ell_1^a\ell_2^b\ell_3^c),\label{nahm}
\end{equation}
where $\ell_1,\ell_2,\ell_3\colon\R^n\rightarrow\R$ are linear forms, $J\in {\rm{Mat}}_{n \times n} (\R)$  is skew-symmetric and $a,b,c\in\R.$ System (\ref{nahm}) has
$$
H=\ell_1^a\ell_2^b\ell_3^c
$$
as integral of motion and
\begin{equation}
\Omega=\frac{\mathrm dx_1\wedge\mathrm dx_2\wedge\dotsb\wedge\mathrm dx_n}{\ell_1\ell_2\ell_3}\label{nahm_measure}
\end{equation}
as invariant measure form. Note that if $(a,b,c)=(1,1,1)$
and $n$ is an even number then (\ref{nahm}) is a canonical Hamiltonian system on $\R^{n}$ with a homogeneous cubic Hamiltonian.

It is immediate to verify that, if $n=2$ and $J\in {\rm{Mat}}_{2 \times 2} (\R)$  is the canonical symplectic matrix,
the first system in (\ref{k3}) is obtained if $(a,b,c)=(1,1,1)$
and $(\ell_1,\ell_2,\ell_3)=(x_2,\sqrt{3}x_1-x_2,\sqrt{3}x_1+x_2)$, the second system in (\ref{k3}) is obtained if $(a,b,c)=(1,1,2)$
and $(\ell_1,\ell_2,\ell_3)=(x_2,2x_1+3x_2,x_1-x_2)$ and the third system in (\ref{k3}) is obtained if $(a,b,c)=(1,2,3)$
and $(\ell_1,\ell_2,\ell_3)=(x_2,3x_1 -x_2,4x_1 +x_2)$.

We now turn our attention to the Kahan-Hirota-Kimura discretization of the quadratic vector field (\ref{nahm}).
Here and below we avoid to write explicitly the discrete equations of motion discretizing the continuous vector fields we are going to consider. Such equations are constructed by using (\ref{KHKdiscretization}) and the explicit form of the Kahan-Hirota-Kimura map is given by
(\ref{eq: Phi gen}) (or (\ref{eq: Phi gen2}) if $f$ is homogeneous).

We give the following statement (note that 
$(a,b,c)$ denotes any permutation of $\{a,b,c\}$).

\begin{prop}
\label{prop:nahm_measure}
The following claims are true.

\begin{enumerate}

\item[(i)] The Kahan-Hirota-Kimura map discretizing system (\ref{nahm})
admits (\ref{nahm_measure}) as invariant measure form.

\smallskip

\item[(ii)] Let $(a,b,c)=(1,1,2)$ and assume that $\ell_3=\alpha \ell_1+\beta \ell_2$ for $\alpha,\beta\in\R$. Then the corresponding Kahan-Hirota-Kimura map 
admits the conserved quantity
\begin{equation}
H_{\epsilon}=\frac{H}{1+\epsilon^2P_2+\epsilon^4P_4} , \label{nahm112_Heps}
\end{equation}
where 
\begin{eqnarray}
P_2&=&-2 \lambda_{12}^2\left(5\alpha^2\ell_1^2+6\alpha\beta\ell_1\ell_2+5\beta^2\ell_2^2\right), \nonumber\\
P_4&=&\lambda_{12}^4 \left(9\alpha^4\ell_1^4-2\alpha^2\beta^2\ell_1^2\ell_2^2+9\beta^4\ell_2^4\right), \nonumber
\end{eqnarray}
with $\lambda_{12}=\nabla\ell_1^{\T}J\nabla\ell_2$.

\smallskip

\item[(iii)] Let $(a,b,c)=(1,2,3)$ and assume that $\ell_3=\alpha \ell_1+\beta \ell_2$ for $\alpha,\beta\in\R$. Then the corresponding Kahan-Hirota-Kimura map 
admits the conserved quantity
\begin{equation}
H_{\epsilon}=\frac{H}{1+\epsilon^2P_2+\epsilon^4P_4+\epsilon^6P_6}, \label{nahm123_Heps}
\end{equation}
where 
\begin{eqnarray}
\qquad P_2&=&-\lambda_{12}^2\left(32\alpha^2\ell_1^2+40\alpha\beta\ell_1\ell_2+35\beta^2\ell_2^2\right), \nonumber \\
\qquad P_4&=&\lambda_{12}^4\left(256\alpha^4\ell_1^4+640\alpha^3\beta\ell_1^3\ell_2+
960\alpha^2\beta^2\ell_1^2\ell_2^2 +592\alpha\beta^3\ell_1\ell_2^3 
+259\beta^4\ell_2^4\right), \nonumber\\
\qquad P_6&=&-\lambda_{12}^6\left(-1152\alpha^3\beta^3\ell_1^3\ell_2^3-864\alpha^2\beta^4\ell_1^2\ell_2^4
-
216\alpha\beta^5\ell_1\ell_2^5+225\beta^6\ell_2^6\right), \nonumber
\end{eqnarray}
with $\lambda_{12}=\nabla\ell_1^{\T}J\nabla\ell_2$.
\end{enumerate}
\end{prop}

\begin{proof} We prove all claims.

\smallskip

\noindent (i) Equations (\ref{nahm}) 
can be written as
\begin{equation}
\dot x=J\left(a\ell_2\ell_3\nabla\ell_1+b\ell_1\ell_3\nabla\ell_2+c\ell_1\ell_2\nabla\ell_3\right).
\label{12}
\end{equation}
Their Kahan-Hirota-Kimura discretization reads
\begin{equation}
\label{nahm_HK}
\x-x= \ep J\Big(a(\ell_2\t \ell_3+\t \ell_2\ell_3)\nabla\ell_1+b(\ell_1 \t \ell_3+\t \ell_1\ell_3)\nabla\ell_2
 +c(\ell_1\t \ell_2+\t \ell_1\ell_2 )\nabla\ell_3\Big).
\end{equation}
Here and below we adopt the shortcut notation $\t {F} = F(\x)$ for any map $F\colon\R^n\rightarrow V$, where $V=\R$, $\R^n$ or ${\rm{Mat}}_{n \times n} (\R)$.
Now, multiplying (\ref{nahm_HK}) from the left by the vectors $\nabla\ell_i^{\T}$, $i=1,2,3$, we obtain
\begin{eqnarray}
\t \ell_1 -\ell_1&=&\ep \lambda_{12}b(\ell_1 \t \ell_3+\t \ell_1\ell_3)+\ep\lambda_{13}c(\ell_1\t \ell_2+\t \ell_1\ell_2),
\label{nahm_e1}\\
\t \ell_2 -\ell_2&=&-\ep\lambda_{12}a(\ell_2 \t \ell_3+\t \ell_2\ell_3)+\ep\lambda_{23}c(\ell_1 \t \ell_2+\t \ell_1\ell_2),
\label{nahm_e2}\\
\t \ell_3 -\ell_3&=&-\ep\lambda_{13}a(\ell_2\t \ell_3+\t \ell_2\ell_3)-\ep\lambda_{23}b(\ell_1 \t \ell_3+\t \ell_1\ell_3),
\label{nahm_e3}
\end{eqnarray}
where $\lambda_{ij}=\nabla\ell_i^{\T}J\nabla\ell_j$, for $1\leq i<j\leq 3$.

Now, the Jacobian of the vector field (\ref{12}) is
\begin{equation}
f'=J(A_1\nabla\ell_1^{\T}+A_2\nabla\ell_2^{\T}+A_3\nabla\ell_3^{\T}),
\nn
\end{equation}
where 
$$
A_1=b\ell_3\nabla\ell_2+c\ell_2\nabla\ell_3,
\qquad
A_2=a\ell_3\nabla\ell_1+c\ell_1\nabla\ell_3,
\qquad
A_3=a\ell_2\nabla\ell_1+b\ell_1\nabla\ell_2.
$$
As for any Kahan-Hirota-Kimura discretization we have
\begin{equation}
\nn
\det \left( \frac{\partial \x}{\partial x }\right)=\frac{\det(I+\epsilon \t f')}{\det(I-\epsilon f')}.
\end{equation}
Using Sylvester's determinant formula we obtain
\begin{equation}
\det(I-\epsilon f')=(1-\epsilon\nabla\ell_1^{\T}JA_1)(1-\epsilon\nabla\ell_2^{\T}B_2JA_2)
(1-\epsilon\nabla\ell_3^{\T}B_1JA_3),
\label{1}
\end{equation}
where 
$$
B_1=\left(I-\epsilon J(A_1\nabla\ell_1^{\T}+A_2\nabla\ell_2^{\T})\right)^{-1},\qquad
B_2=\left(I-\epsilon J A_1\nabla\ell_1^{\T}\right)^{-1},\nonumber
$$
or, more explicitly (use the Sherman-Morrison formula),
\begin{align}
B_1&=I+\epsilon J(\eta_{11}A_1\nabla\ell_1^{\T}+\eta_{12}A_1\nabla\ell_2^{\T}
+\eta_{21}A_2\nabla\ell_1^{\T}+\eta_{22}A_2\nabla\ell_2^{\T}),\nn \\
B_2&=I+\frac{\epsilon JA_1\nabla\ell_1^{\T}}{1-\epsilon \nabla\ell_1^{\T}JA_1},
\nn
\end{align}
where
\begin{align}
\eta_{11}&=\frac{1-\epsilon\nabla\ell_2^{\T}JA_2}{\Delta},
&\eta_{12}&=\frac{\epsilon\nabla\ell_1^{\T}JA_2}{\Delta},\nn \\
\eta_{21}&=\frac{\epsilon\nabla\ell_2^{\T}JA_1}{\Delta},
&\eta_{22}&=\frac{1-\epsilon\nabla\ell_1^{\T}JA_1}{\Delta}, \nn
\end{align}
with
\begin{equation}
\Delta=1-\epsilon(\nabla\ell_1^{\T}JA_1+\nabla\ell_2^{\T}JA_2)
+\epsilon^2(\nabla\ell_1^{\T}JA_1\nabla\ell_2^{\T}JA_2-\nabla\ell_1^{\T}JA_2\nabla\ell_2^{\T}JA_1).\nn
\end{equation}
Replacing $(x,\epsilon)$ by $(\x,-\epsilon)$ in (\ref{1}) we obtain $\det(I+\epsilon \t f')$. 
Note that the expressions for $\det(I-\epsilon f')$ and  $\det(I+\epsilon \t f')$ are rational functions in the variables $\epsilon,\ell_1,\ell_2,\ell_3$ and $\epsilon,\t \ell_1,\t \ell_2,\t \ell_3$ respectively, i.e.,
$
 \det(I+\epsilon \t f')=\widetilde{P}/\widetilde{Q}$
 and $
 \det(I-\epsilon f')=P/Q
$
for polynomials $\widetilde{P},\widetilde{Q}\in\Q[\epsilon,\t \ell_1,\t \ell_2,\t \ell_3]$ and ${P},{Q}\in\Q[\epsilon,\ell_1,\ell_2,\ell_3]$. 
 Now, using a computer algebra program like Maple it can be checked that the polynomial
$
\widetilde{P}Q\ell_1\ell_2\ell_3-P\widetilde{Q} \, \t \ell_1\t \ell_2\t \ell_3
$
is contained in the ideal $I\subset\Q[\epsilon,\ell_1,\ell_2,\ell_3,\t \ell_1,\t \ell_2,\t \ell_3]$ generated by (\ref{nahm_e1}--\ref{nahm_e3}), i.e.,
$$
\frac{\det(I+\epsilon \t f')}{\det(I-\epsilon f')}=\frac{\t \ell_1\t \ell_2\t \ell_3}{\ell_1\ell_2\ell_3}
$$
is an algebraic identity. This proves the claim.

\smallskip

\noindent (ii) In this case we find that the discrete equations of motion can be written as
\beq
\label{nahm113_HK}
\x-x=\epsilon J\left((3\alpha(\ell_1\t \ell_2+\t \ell_1\ell_2)+2\beta\ell_2\t \ell_2)\nabla\ell_1+(3\beta(\ell_1\t \ell_2+\t \ell_1\ell_2)+2\alpha\ell_1\t \ell_1)\nabla\ell_2\right).
\eeq
Then multiplying (\ref{nahm113_HK}) from the left by $\nabla\ell_i^{\T}$, $i=1,2$, we obtain
\begin{align}
\t \ell_1-\ell_1&=\epsilon\lambda_{12}(3\beta(\ell_1\t \ell_2+\t \ell_1\ell_2)+2\alpha\ell_1\t \ell_1),\label{nahm112_e1}\\
\t \ell_2-\ell_2&=-\epsilon\lambda_{12}(3\alpha(\ell_1\t \ell_2+\t \ell_1\ell_2)+2\beta\ell_2\t \ell_2).\label{nahm112_e2}
\end{align}
Now, using (\ref{nahm112_e1}), (\ref{nahm112_e2}) it can be verified that (\ref{nahm112_Heps}) is indeed a conserved quantity of the Kahan-Hirota-Kimura map.

\smallskip

\noindent (iii) In this case we find that the discrete equations of motion can be written as
\beq
\label{nahm112_HK}
\x-x=\epsilon J\left((4\alpha(\ell_1\t \ell_2+\t \ell_1\ell_2)+2\beta\ell_2\t \ell_2)\nabla\ell_1+(5\beta(\ell_1\t \ell_2+\t \ell_1\ell_2)+4\alpha\ell_1\t \ell_1)\nabla\ell_2\right).
\eeq
Then multiplying (\ref{nahm112_HK}) from the left by $\nabla\ell_i^{\T}$, $i=1,2$, we obtain
\begin{align}
\t \ell_1-\ell_1&=\epsilon\lambda_{12}(5\beta(\ell_1\t \ell_2+\t \ell_1\ell_2)+4\alpha\ell_1\t \ell_1),\label{nahm123_e1}\\
\t \ell_2-\ell_2&=-\epsilon\lambda_{12}(4\alpha(\ell_1\t \ell_2+\t \ell_1\ell_2)+2\beta\ell_2\t \ell_2).\label{nahm123_e2}
\end{align}
Now, using (\ref{nahm123_e1}), (\ref{nahm123_e2}) it can be verified that (\ref{nahm123_Heps}) is indeed a conserved quantity of the Kahan-Hirota-Kimura map.
\end{proof}

We remark that there exist at least other two remarkable choices for the parameters $a,b,c$. The first one is give by $(a,b,c)=(1,1,1)$. The Kahan-Hirota-Kimura map in this case admits a conserved quantity as well, but such a case is covered by Proposition 3 in
\cite{C1}. The second one is $(a,b,c)=(1,1,-1)$. We will see that such case is contained in Proposition \ref{prop:1}.

\section{Higher-dimensional Nambu systems}

An important class of three-dimensional vector fields is given by Nambu
systems \cite{nambu}:
\beq
\dot x =\nabla H \times \nabla K,
\label{fd}
\eeq
where $H,K\colon\R^n\rightarrow\R$. In these coordinates one has a divergenceless vector field, that means that the canonical
measure is preserved by the flow. The famous Euler top
is an example of Nambu mechanics.
The integrability properties of the Kahan-Hirota-Kimura discretization of Nambu systems in $\R^3$, with $H,K$ being quadratic polynomials, has been studied in \cite{C2} and \cite{HP}.

In this paper we consider the following superintegrable $n$-dimensional generalization of (\ref{fd}).
Let $H,K\colon\R^n\rightarrow\R$ be two functionally independent quadratic homogeneous polynomials and
$v_1,\dotsc,v_{n-3}\in\R^n$ be $n-3$ linearly independent vectors in $\R^n$.
Then define the divergenceless system of quadratic differential equations given by
\begin{equation}
\label{vf_gen_nabmu}
\dot x=\det(\e,v_1,\dotsc,v_{n-3},\nabla H,\nabla K),
\end{equation}
where $\e=(\e_1,\dotsc,\e_n)^{\T}$, $\e_i$ being the $i$-th unit basis vector in $\R^n$. Equivalently, one has
$$
\dot x =\left( \det A_1, \dots, (-1)^{i+1}\det A_i, \dots, (-1)^{n+1}\det A_n \right),
$$
 where $A_i$ denotes the matrix obtained from $(v_1,\dotsc,v_{n-3},\nabla H,\nabla K)$ by deleting the $i$-th row. It is easy to see that (\ref{vf_gen_nabmu})
 is a superintegrable system. In the generic case (i.e. $\dot x\nequiv 0$), it admits $H$, $K$ and 
 \beq
V_i= \sum_{k=1}^n (v_i)_k x_k, \qquad i=1,\dots,n-3. \label{ff}
\eeq
as $n-1$ functionally independent integrals of motion.

 \begin{ex}
The periodic Volterra chain with $N=4$ is governed by the following system of differential equations:
\begin{equation}
\label{vf_PVC4}
\left\{ \begin{array}{l} 
\dot x_1=x_1(x_2-x_4),\vspace{.1truecm}\\
\dot x_2=x_2(x_3-x_1),\vspace{.1truecm}\\
\dot x_3=x_3(x_4-x_2),\vspace{.1truecm}\\
\dot x_4=x_4(x_1-x_3).
\end{array} \right.
\end{equation}
This system is completely integrable and possesses three functionally independent integrals of motion: $H=x_1x_3$ and $K=x_2x_4$ and $V_1=x_1+x_2+x_3+x_4$.
The vector field (\ref{vf_PVC4}) can be expressed as 
\begin{equation*}
\dot x=\det(\e,v_1,\nabla H,\nabla K),
\end{equation*}
where $v_1=\nabla V_1$ and  $\e=(\e_1,\dotsc,\e_4)^{\T}$.
\end{ex}

We now turn our attention to the Kahan-Hirota-Kimura
discretization of system (\ref{vf_gen_nabmu}).

\begin{prop}
\label{prop:4}
The following claims are true.
\begin{enumerate}

\item[(i)] Assume that $H=\ell_1 \ell_2$, where $\ell_1,\ell_2\colon\R^n\rightarrow\R$ are linear forms. Then the
Kahan-Hirota-Kimura map discretizing system (\ref{vf_gen_nabmu}) admits
$$
H_{\epsilon}=\frac{\ell_1\ell_2}{1-\epsilon^2\Delta^2},
$$
where $\Delta=\det(\nabla \ell_1,v_1,\dotsc,v_{n-3},\nabla \ell_2,\nabla K)$ and $V_i$, $i=1,\dots,n-3$, given in (\ref{ff}), as $n-2$ functionally independent conserved quantities.

\smallskip

\item[(ii)] Assume that $H=\ell_1 \ell_2$, $K=\ell_3 \ell_4$, where $\ell_1,\ell_2,\ell_3,\ell_4 \colon\R^n\rightarrow\R$ are linear forms. Then the
Kahan-Hirota-Kimura map discretizing system (\ref{vf_gen_nabmu}) admits
$$
H_{\epsilon}=\frac{\ell_1\ell_2}{1-\epsilon^2(\ell_3\Delta_{124}+\ell_4\Delta_{123})^2}, \qquad
K_\ep=\frac{\ell_3\ell_4}{1-\epsilon^2(\ell_1\Delta_{234}+\ell_2\Delta_{134})^2},
$$
where $\Delta_{ijk}=\det(\nabla \ell_i,v_1,\dotsc,v_{n-3},\nabla \ell_j,\nabla \ell_k)$
and $V_i$, $i=1,\dots,n-3$, given in (\ref{ff}), as $n-1$ functionally independent  conserved quantities. In such case the Kahan-Hirota-Kimura map admits the invariant measure form
\begin{equation*}
\Omega=\frac{\mathrm dx_1\wedge\mathrm dx_2\wedge\dotsb\wedge\mathrm dx_n}{\ell_1\ell_2\ell_3\ell_4}.
\end{equation*}

\end{enumerate}
\end{prop}

\begin{proof} We prove both claims.

\smallskip

\noindent (i) 
The Kahan-Hirota-Kimura map is given by
\beq
\label{HKf_4}
\begin{split}
\x-x=&\epsilon\left(\det\left(\e,v_1,\dotsc,v_{n-3},\ell_2\nabla \ell_1+\ell_1\nabla \ell_2,\t {\nabla K}\right)\right.\\&\left.+\det\left(\e,v_1,\dotsc,v_{n-3},\t \ell_2\nabla \ell_1+\t \ell_1\nabla \ell_2,\nabla K\right)\right).
\end{split}
\eeq
Now, multiplying (\ref{HKf_4}) from the left by $\nabla \ell_i^{\T}$, $i=1,2$, yields
\begin{align*}
\t \ell_1(1-\epsilon\Delta)=\ell_1(1+\epsilon\t \Delta),\\
\t \ell_2(1+\epsilon\Delta)=\ell_2(1-\epsilon\t \Delta),
\end{align*}
which justifies the claim. Concerning functions (\ref{ff}), we note that the Kahan-Hirota-Kimura
discretization preserves all linear integrals.

\smallskip

\noindent (ii) 
In this we find that case the discrete equations of motion can be written as
\beq
\label{HKf_5}
\begin{split}
\frac{\x-x}{\epsilon}=&(\ell_1\t \ell_3+\t \ell_1\ell_3)\Delta_{24}+(\ell_1\t \ell_4+\t \ell_1\ell_4)\Delta_{23}+
(\ell_2\t \ell_3+\t \ell_2\ell_3)\Delta_{14}+(\ell_2\t \ell_4+\t\ell_2\ell_4)\Delta_{13},
\end{split}
\eeq
where $\Delta_{ij}=\det(\e,v_1,\dotsc,v_{n-3},\nabla\ell_i,\nabla\ell_j)$.
Then multiplying (\ref{HKf_5}) from the left by the vectors $\nabla \ell_i^{\T}$, $i=1,\dotsc,4$, yields
\begin{align}
\t \ell_1(1-\epsilon(\ell_3\Delta_{124}+\ell_4\Delta_{123}))
&=\ell_1(1+\epsilon(\t \ell_3\Delta_{124}+\t \ell_4\Delta_{123})),\label{id_l1_2}\\
\t \ell_2(1+\epsilon(\ell_3\Delta_{124}+\ell_4\Delta_{123}))
&=\ell_2(1-\epsilon(\t \ell_3\Delta_{124}+\t \ell_4\Delta_{123})),\label{id_l2_2}\\
\t \ell_3(1+\epsilon(\ell_1\Delta_{234}+\ell_2\Delta_{134}))
&=\ell_3(1-\epsilon(\t \ell_1\Delta_{234}+\t \ell_2\Delta_{134})),\label{id_l3_2}\\
\t \ell_4(1-\epsilon(\ell_1\Delta_{234}+\ell_2\Delta_{134}))
&=\ell_4(1+\epsilon(\t \ell_1\Delta_{234}+\t \ell_2\Delta_{134})).\label{id_l4_2}
\end{align}

Computing the Jacobian of the vector field (\ref{vf_gen_nabmu}) we obtain
\begin{equation*}
f'=A_1\nabla\ell_1^{\T}+A_2\nabla\ell_2^{\T}+A_3\nabla\ell_3^{\T}+A_4\nabla\ell_4^{\T},
\end{equation*}
where the vectors $A_i$, $i=1,\dotsc,4$ are defined by
\begin{align*}
A_1&=\ell_3\Delta_{24}+\ell_4\Delta_{23}, &A_2&=\ell_3\Delta_{14}+\ell_4\Delta_{13},\\
A_3&=\ell_1\Delta_{24}+\ell_2\Delta_{14}, &A_4&=\ell_1\Delta_{23}+\ell_2\Delta_{13}.
\end{align*}
As for any Kahan-Hirota-Kimura discretization we have
\begin{equation}
\nn
\det \left( \frac{\partial \x}{\partial x }\right)=\frac{\det(I+\epsilon \t f')}{\det(I-\epsilon f')}.
\end{equation}
Using Sylvester's determinant formula we obtain
\begin{equation*}
\det(I-\epsilon f')=(1-\epsilon\nabla\ell_1^{\T}A_1)(1-\epsilon\nabla\ell_2^{\T}B_3A_2) (1-\epsilon\nabla\ell_3^{\T}B_2A_3)(1-\epsilon\nabla\ell_4^{\T}B_1A_4),
\end{equation*}
where
\begin{gather*}
B_1=\left(I-\epsilon(A_1\nabla\ell_1^{\T}+A_2\nabla\ell_2^{\T}+A_3\nabla\ell_3^{\T})\right)^{-1},\\
B_2=\left(I-\epsilon(A_1\nabla\ell_1^{\T}+A_2\nabla\ell_2^{\T})\right)^{-1},\quad
B_3=\left(I-\epsilon A_1\nabla\ell_1^{\T}\right)^{-1}.
\end{gather*}
As in the proof of Proposition \ref{prop:nahm_measure} the matrices $B_i$, $i=1,2,3$, can be computed using the Sherman-Morrison formula. 
Replacing $(x,\epsilon)$ by $(\x,-\epsilon)$ we obtain $\det(I+\epsilon \t f')$. Note that the expressions for $\det(I-\epsilon f')$ and $\det(I+\epsilon \t f')$ are rational functions in the variables $\epsilon,\ell_1,\dotsc,\ell_4$ and $\epsilon,\t \ell_1,\dotsc\t \ell_4$ respectively. Finally, using the equations 
(\ref{id_l1_2}--\ref{id_l4_2}) it can be verified that
\begin{equation*}
\frac{\det(I+\epsilon \t f')}{\det(I-\epsilon f')}=\frac{\t \ell_1\t \ell_2\t \ell_3\t \ell_4}{\ell_1\ell_2\ell_3\ell_4}
\end{equation*}
is an algebraic identity. This proves the claim.
\end{proof}

\section{Classes with exact preservation of integrals}

This section is devoted to two families of systems of quadratic ordinary differential equations in $\R^n$ which admit a rational conserved quantity. We will show that the corresponding  Kahan-Hirota-Kimura discretizations
preserve exactly such integrals of motion.

The first family of differential equations is given by
\begin{equation}
\label{vf_1}
\dot x=\ell^2J\nabla\left(\frac{H}{\ell}\right),
\end{equation}
where $\ell\colon\R^n\rightarrow\R$ is a linear
form, $H\colon\R^n\rightarrow\R$ is a quadratic homogeneous polynomial and  $J\in {\rm{Mat}}_{n \times n} (\R)$ is skew-symmetric. System (\ref{vf_1}) admits
the rational function
\begin{equation}
\label{F1}
F=\frac{H}{\ell}
\end{equation}
as integral of motion.

The second family of differential equations we want to consider is given by
\begin{equation}
\label{vf_2}
\dot x=\det\left(\e,v_1,\dotsc,v_{n-3},\ell_2^2\nabla\left(\frac{\ell_1}{\ell_2}\right),\nabla H\right) - \alpha H \det(\nabla\ell_1,v_1,\dotsc,v_{n-3},\nabla\ell_2,\e),
\end{equation}
with parameter $\alpha\in\R$,
$\ell_1,\ell_2\colon\R^n\rightarrow\R$ linear forms, $H\colon\R^n\rightarrow\R$ quadratic homogeneous polynomial and $v_1,\dotsc,v_{n-3}\in\R^n$ are constant vectors such that $\nabla \ell_1,\nabla \ell_2,v_1,\dotsc,v_{n-3}$ are linearly independent. 
System (\ref{vf_2}) admits
the rational function
\begin{equation}
\label{F2}
F=\frac{\ell_1}{\ell_2}
\end{equation}
as integral of motion. Moreover the linear forms
\begin{equation}
\label{F3}
V_i = \sum_{k=1}^n (v_i)_k x_k, \qquad i=1,\dots,n-3,
\end{equation}
provide $n-3$ additional functionally independent integrals of motion. Therefore system (\ref{vf_2}) is superintegrable.
Additionally, if $\alpha=1$, system (\ref{vf_2}) admits the rational function 
\begin{equation}
\label{F4}
G=\frac{H}{\ell_1}
\end{equation}
as integral of motion.

\begin{ex}
A remarkable three-dimensional example belonging to class (\ref{vf_2})
is given by
\beq
\left\{ \begin{array}{l}
\dot{x}_1= \displaystyle{\frac12(a_1+a_2+a_3)x_1^2 +\frac12 (a_1-a_2-a_3)(x_1-x_2)(x_1-x_3)}, \vspace{.2truecm} \\
\dot{x}_2= \displaystyle{\frac12(a_1+a_2+a_3)x_2^2 +\frac12 (a_2-a_3-a_1)(x_2-x_3)(x_2-x_1)}, \vspace{.2truecm} \\
\dot{x}_3=\displaystyle{\frac12(a_1+a_2+a_3)x_3^2 +\frac12 (a_3-a_1-a_2)(x_3-x_1)(x_3-x_2)},
\end{array} \right.
\label{h3bid}
\eeq
where $a_1,a_2,a_3$ are  parameters. 
Indeed, system
(\ref{h3bid}) is a degenerate subcase of the famous Halphen system (1881) \cite{H1}:
\beq
\left\{ \begin{array}{l}
\dot{x}_1= a_1 x_1^2+ (\lambda-a_1)
(x_1x_2-x_2x_3 +x_3x_1), \vspace{.2truecm} \\
\dot{x}_2= a_2 x_2^2+ (\lambda-a_2)
(x_2x_3-x_3x_1 +x_1x_2), \vspace{.2truecm} \\
\dot{x}_3= a_3 x_3^2+ (\lambda-a_3)
(x_3x_1-x_1x_2 +x_2x_3),
\end{array} \right.
\label{h1}
\eeq
where $\lambda$ is an additional parameter. 
The general solution of (\ref{h1}) 
can be written in terms of hypergeometric functions
provided that
$
\lambda\neq 0, \,a_1+a_2+a_3\neq 2\lambda
$ \cite{H1}.
It turns out that (\ref{h3bid}) comes from (\ref{h1}) 
with the choice
$\lambda=(a_1+a_2+a_3)/2$ with 
$a_1+a_2+a_3\neq0$, which evidently violates the previous condition. Obviously, such choice renders 
system (\ref{h3bid}) much simpler than the original Halphen system (\ref{h1}).

By straightforward inspection one can see that
(\ref{h3bid}) admits two functionally independent rational
integrals of motion of the form
\beq
F_{ij}=\frac{a_1(x_1x_2-x_2x_3 +x_3x_1)+a_2(x_2x_3-x_3x_1 +x_1x_2)+a_3(x_3x_1-x_1x_2 +x_2x_3)}{x_i-x_j}, \label{sfgfh}
\eeq
where $i,j=1,2,3$, $i\neq j$. Incidentally we observe that system 
(\ref{h3bid}) has been explicitly considered in \cite{Ma}. Here the author claims that the parameter-independent quantities
$(x_1-x_2)/(x_2-x_3)$ and $(x_2-x_3)/(x_3-x_1)$
are integrals of motion of (\ref{h3bid}). Although such a  statement is correct, these integrals are functionally dependent, so that nothing can be argued about the complete integrability of the system. On the contrary, the existence of two independent integrals of motion (\ref{sfgfh}) allows one to find a (rank 2) bi-Hamiltonian structure of (\ref{h3bid}), thus proving its complete integrability in the Arnold-Liouville sense. A representative polynomial member of such a Poisson pencil is the following one:
\beq \label{sfgdh} 
\{x_1,x_2\}=(x_1-x_2)(x_2-x_3), \;
\{x_2,x_3\}=(x_2-x_3)^2,\;
\{x_3,x_1\}=(x_2-x_3)(x_3-x_1). 
\eeq
It turns out that (\ref{h3bid}) is Hamiltonian with respect to (\ref{sfgdh}) with Hamiltonian $F_{23}/2$ while $F_{23}/F_{31}=(x_3-x_1)/(x_2-x_3)$  is a Casimir function.

We conclude by saying that system (\ref{h3bid}) belongs to the family of vector fields (\ref{vf_2}) by setting $n=3$,
$\ell_1=x_3-x_1$, $\ell_2=x_2-x_3$, $H=(x_1-x_2)F_{12}$ and $\alpha=1$.
\end{ex}

The following statements show that the 
Kahan-Hirota-Kimura maps discretizing systems
(\ref{vf_1}) and (\ref{vf_2}) preserve exactly the 
continuous integrals of motion.

\begin{prop}
\label{prop:1}
The Kahan-Hirota-Kimura discretization of system
(\ref{vf_1}) admits (\ref{F1}) as conserved quantity.
\end{prop}

\begin{proof}
Here and below we use the notation $\Hxx=x^{\T}(\nabla^2H)x$, $\HXx=\x^{\T}(\nabla^2H)x$ and $\HXX=\x^{\T}(\nabla^2H)\x$, where
$\nabla^2H$ denotes the Hessian of a function $H\colon\R^n\rightarrow\R$.

The Kahan-Hirota-Kimura map is given by
\beq
\label{HKf_1}
\x-x=\epsilon\left(\ell J\t {\nabla H}+\t \ell J\nabla H-\HXx J\nabla \ell\right).
\eeq
Now, multiplying (\ref{HKf_1}) from the left by $\ell \t {\nabla H^{\T}}$, $\t \ell\nabla H^{\T}$ and $\nabla \ell^{\T}$ we obtain
\begin{align}
\ell\HXX-\ell\HXx&=\epsilon\left(\ell\t \ell\t {\nabla H^{\T}}J\nabla H-\ell\HXx\t {\nabla H^{\T}}J\nabla \ell\right),\label{B}\\
\t \ell\HXx-\t \ell\Hxx&=\epsilon\left(\ell\t \ell\nabla H^{\T}J\t {\nabla H}-\t \ell\HXx\nabla H^{\T}J\nabla \ell\right),\label{A}\\
\t \ell-\ell&=\epsilon\left(\ell\nabla \ell^{\T}J\t {\nabla H}+\t \ell\nabla \ell^{\T}J\nabla H\right).\label{C}
\end{align}
Finally, adding (\ref{A}) to (\ref{B}) and substituting (\ref{C}) we obtain $\ell\HXX=\t \ell\Hxx$, which gives the claim.
\end{proof}

\begin{prop}
\label{prop:2}
The following claims are true.

\begin{enumerate}

\item[(i)] The Kahan-Hirota-Kimura discretization of system
(\ref{vf_2}) admits (\ref{F2}) and (\ref{F3}) as $n-2$ functionally independent conserved quantities.

\smallskip

\item[(ii)] The Kahan-Hirota-Kimura discretization of system
(\ref{vf_2}) with $\alpha=1$ admits (\ref{F4}) as additional conserved quantity.

\end{enumerate}
\end{prop}

\begin{proof} We prove all claims.

\smallskip

\noindent (i) The Kahan-Hirota-Kimura map is given by
\begin{eqnarray}
\frac{\x-x }{\epsilon}&=&\det(\e,v_1,\dotsc,v_{n-3},\t \ell_2\nabla\ell_1-\t \ell_1\nabla\ell_2,\nabla H) \label{HKf_2}\\
&&+\det(\e,v_1,\dotsc,v_{n-3},\ell_2\nabla\ell_1-\ell_1\nabla\ell_2,\t {\nabla H}) - \alpha\HXx \det(\nabla\ell_1,v_1,\dotsc,v_{n-3},\nabla\ell_2,\e). \nonumber
\end{eqnarray}
Now, multiplying (\ref{HKf_2}) from the left by $\nabla \ell_i^{\T}$, $i=1,2$, we obtain
\begin{align}
\t \ell_1-\ell_1 =& -\epsilon\left(\t \ell_1\Delta + \ell_1\t \Delta\right)\label{HKf_2_1},\\
\t \ell_2-\ell_2 =& -\epsilon\left(\t \ell_2\Delta + \ell_2\t \Delta\right)\label{HKf_2_2},
\end{align}
where $\Delta=\det(\nabla\ell_1,v_1,\dotsc,v_{n-3},\nabla\ell_2,\nabla H)$.
Thus, we see that $\t \ell_1\ell_2=\t \ell_2 \ell_1$.

\noindent (ii) Let $\alpha=1$.
Multiplying (\ref{HKf_2}) from the left by the vectors $\ell_1\t {\nabla H^{\T}}$ and $\t \ell_1\nabla H^{\T}$ we obtain
\begin{align}
\ell_1\HXX-\ell_1\HXx =& \epsilon\left(-\ell_1\det(\nabla H,v_1,\dotsc,v_{n-3},\t \ell_2\nabla\ell_1-\t \ell_1\nabla\ell_2,\t {\nabla H})-\ell_1\HXx\t {\Delta}\right),\label{HKf_2_3}\\
\t \ell_1\HXx-\t \ell_1\Hxx =& \epsilon\left(\t \ell_1\det(\nabla H,v_1,\dotsc,v_{n-3},\ell_2\nabla\ell_1-\ell_1\nabla\ell_2,\t {\nabla H})-\t \ell_1\HXx\Delta\right)\label{HKf_2_4}.
\end{align}
Adding (\ref{HKf_2_3}), (\ref{HKf_2_4}) and substituting (\ref{HKf_2_1}), we conclude that $\ell_1\HXX=\t \ell_1\Hxx$.
\end{proof}

The next claim follows.

\begin{cor}
Assume that $\ell_1,\ell_2,\ell_3,\ell_4\colon\R^n\rightarrow\R$ are linear forms.
Then the Kahan-Hirota-Kimura map discretizing 
\begin{equation}
\label{vf_4}
\dot x=\det\left(\e,v_1,\dotsc,v_{n-3},\ell_2^2\nabla\left(\frac{\ell_1}{\ell_2}\right),\ell_4^2\nabla\left(\frac{\ell_3}{\ell_4}\right)\right),
\end{equation}
where $\e=(\e_1,\dotsc,\e_n)^{\T}$, admits $F_1=\ell_1/\ell_2$, $F_2=\ell_3/\ell_4$ and (\ref{F3}) as $n-1$ functionally independent conserved quantities. Moreover the Kahan-Hirota-Kimura map has the invariant measure form
\begin{equation}
\Omega=\frac{\mathrm dx_1\wedge\mathrm dx_2\wedge\dotsb\wedge\mathrm dx_n}{\ell_1\ell_2\ell_3\ell_4}.\label{vf_4_measure}
\end{equation}
\end{cor}

\begin{proof}
The conserved quantities follow from Proposition \ref{prop:2}. Indeed, we have
\beq
\label{HKf_3}
\begin{split}
\frac{\x-x}{\epsilon}=&(\ell_1\t \ell_3+\t \ell_1\ell_3)\Delta_{24}-(\ell_1\t \ell_4+\t \ell_1\ell_4)\Delta_{23}-(\ell_2\t \ell_3+\t \ell_2\ell_3)\Delta_{14}+(\ell_2\t \ell_4+\t \ell_2\ell_4)\Delta_{13},
\end{split}
\eeq
where $\Delta_{ij}=\det(\e,v_1,\dotsc,v_{n-3},\nabla \ell_i,\nabla \ell_j)$.
Then multiplying (\ref{HKf_3}) from the left by the vectors $\nabla \ell_i^{\T}$, $i=1,\dotsc,4$, yields
\begin{align}
\t \ell_1(1+\epsilon(\ell_4\Delta_{123}-\ell_3\Delta_{124}))
&=\ell_1(1+\epsilon(\t \ell_3\Delta_{124}-\t \ell_4\Delta_{123})),\nonumber \\
\t \ell_2(1+\epsilon(\ell_4\Delta_{123}-\ell_3\Delta_{124}))
&=\ell_2(1+\epsilon(\t \ell_3\Delta_{124}-\t \ell_4\Delta_{123})),\nonumber \\
\t \ell_3(1+\epsilon(\ell_1\Delta_{234}-\ell_2\Delta_{134}))
&=\ell_3(1+\epsilon(\t \ell_2\Delta_{134}-\t \ell_1\Delta_{234})),\nonumber \\
\t \ell_4(1+\epsilon(\ell_1\Delta_{234}-\ell_2\Delta_{134}))
&=\ell_4(1+\epsilon(\t \ell_2\Delta_{134}-\t \ell_1\Delta_{234})),\nonumber 
\end{align}
where $\Delta_{ijk}=\det(\nabla \ell_i,v_1,\dotsc,v_{n-3},\nabla \ell_j,\nabla \ell_k)$.
The existence of the invariant measure (\ref{vf_4_measure}) can be verified by similar computations as in the proof of Proposition \ref{prop:4}.
\end{proof}

\section{Conclusions}

In this paper we presented some new examples of Kahan-Hirota-Kimura discretizations. The ideal aim of the present work was not only to enrich the already existing long list
of systems for which such discretization scheme preserves good features of the continuous systems (complete integrability in the optimal case). Our goal and hope would be also to attract
attention of experts in the theory of integrable systems and in algebraic geometry. As a matter of fact, the integrability mechanism of the Kahan-Hirota-Kimura
discretization is not unveiled yet. One of the major open problems is indeed to identify and characterize those structural
properties of an integrable continuous vector field which ensure the preservation of integrability in the discrete setting.

\section{Acknowledgements}

We thank Yuri Suris for inspiring discussions.
This research is supported by the DFG Collaborative Research Center TRR 109 ``Discretization in Geometry and Dynamics''.



\begin{thebibliography}{}

\bibitem{CT}
A.S.~Carstea, T.~Takenawa,
{\em A note on minimization of rational surfaces obtained from birational dynamical systems},
J. Nonlin. Math. Phys. {\bf 20} (2013), No.1, 17--33.

\bibitem{C1}
E. Celledoni, R.I. McLachlan, B. Owren, 
G.R.W. Quispel,
{\em Geometric properties of Kahan's method}, 
J. Phys. A {\bf 46} (2013), 025201, 12 pp.

\bibitem{C2}
E. Celledoni, R.I. McLachlan, D.I. McLaren, B. Owren, G.R.W. Quispel,
{\em Integrability properties of Kahan's method}, 
J. Phys. A {\bf{47}} (2014), 365202, 20 pp.

\bibitem{C3}
E. Celledoni, R.I. McLachlan, D.I. McLaren, B. Owren, G.R.W. Quispel,
{\em Discretization of polynomial vector fields by polarization},
Proc. R. Soc. A {\bf 471} (2015), 20150390, 10 pp.

\bibitem{H1}
G.H.~Halphen, 
{\em{Sur certains syst\`emes d'\'equations 
diff\'erentielles}},
C. R. Acad. Sci., Paris {\bf{92}} (1881), 1404--1406.

\bibitem{HK1}
R.~Hirota, K.~Kimura,
{\em Discretization of the Euler top},
J. Phys. Soc. Japan {\bf 69} (2000), No. 3, 627--630.

\bibitem{HMM}
N.J.~Hitchin , N.S.~Manton, M.K.~Murray,
{\em Symmetric monopoles},
Nonlin. {\bf{8}} (1995), No. 5,  661--692.


\bibitem{HP}
A.N.W.~Hone, M. Petrera,
{\em Three-dimensional discrete systems of Hirota-Kimura type and deformed Lie-Poisson algebras},
J. Geom. Mech. {\bf{1}} (2009), No. 1, 55--85. 

\bibitem{K}
W.~Kahan,
{\em Unconventional numerical methods for trajectory calculations},
Unpublished lecture notes, 1993.

\bibitem{HK2}
K.~Kimura, R.~Hirota,
{\em Discretization of the Lagrange top},
J. Phys. Soc. Japan {\bf 69} (2000), No. 10, 3193--3199.

\bibitem{Ko}
T.E.~Kouloukas, G.R.W.~Quispel, P.~Vanhaecke,
{\em Liouville integrability and superintegrability of a generalized Lotka-Volterra system and its Kahan discretization}, J. Phys. A {\bf{49}} (2016), No. 22, 225201, 13 pp.

\bibitem{Ma}
A.J.~Maciejewski,
{\em{About algebraic integrability and non-integrability of ordinary differential equations}},
Chaos, Solitons \& Fractals {\bf{9}} (1998), No.1-2,
 51--65.


\bibitem{nambu}
Y. Nambu,
\emph{Generalized Hamiltonian dynamics},
Phys. Rev. D, {\bf 7} (1973), No.8, 2405--2412.

\bibitem{PPS1}
M.~Petrera, A.~Pfadler, Yu.B.~Suris,
{\em On integrability of Hirota-Kimura-type discretizations: experimental study of the discrete Clebsch system},
Exp. Math. {\bf 18} (2009), No. 2, 223--247.

\bibitem{PPS2}
M. Petrera, A. Pfadler, Yu.B. Suris,
{\em On integrability of Hirota-Kimura type discretizations},
Regular Chaotic Dyn. {\bf 16} (2011), No. 3-4, p. 245--289.

\bibitem{PPS3}
M. Petrera, A. Pfadler, Yu.B. Suris,
{\em On the construction of elliptic solutions of
integrable birational maps},
Exp. Math. (data not yet available).

 

\bibitem{PS1}
M.~Petrera, Yu.B.~Suris,
{\em On the Hamiltonian structure of Hirota-Kimura discretization of the Euler top},
 Math. Nachr. {\bf 283} (2010), No. 11, 1654--1663.
 
 \bibitem{PS2} 
M. Petrera, Yu.B. Suris,
{\em S. V. Kovalevskaya system, its generalization and discretization}, Front. Math. China {\bf{8}} (2013), No. 5, 1047--1065.

\bibitem{PS3} 
M. Petrera, Yu.B. Suris,
{\em A construction of a large family of commuting pairs of integrable symplectic birational 4-dimensional maps},
{\url http://arxiv.org/abs/1606.08238}

\bibitem{PS4} 
M. Petrera, Yu.B. Suris,
{\em A construction of commuting systems of integrable symplectic birational maps}, \url{http://arxiv.org/abs/1607.07085}.

\bibitem{S}
Yu.B.~Suris, {\em The problem of integrable discretization:
Hamiltonian approach}, Progress in Mathematics 219. Basel:
Birkh\"auser, 2003.


 



\end{thebibliography}

\end{document}